\documentclass[a4paper, 11pt]{article}
\usepackage{amsmath, amssymb,amscd, amsthm}
\usepackage[mathscr]{eucal}
\usepackage{graphics}
\usepackage{fullpage}
\usepackage{url}
\newcommand\cyr{%
 \renewcommand\rmdefault{wncyr}%
 \renewcommand\sfdefault{wncyss}%
 \renewcommand\encodingdefault{OT2}%
\normalfont\selectfont} \DeclareTextFontCommand{\textcyr}{\cyr}

\newtheorem{theorem}{Theorem}
\newtheorem{lemma}[theorem]{Lemma}

\newtheorem{definition}[theorem]{Definition}

\begin{document}

\title{\textbf{The Contest Between Simplicity and Efficiency in Asynchronous Byzantine Agreement}}

\author{Allison Lewko \thanks{Supported by a National Defense Science and Engineering Graduate Fellowship.} \\ University of Texas at Austin
\\ \texttt{alewko@cs.utexas.edu}}

\date{}
\maketitle

\begin{abstract}
In the wake of the decisive impossibility result of Fischer, Lynch, and Paterson for deterministic consensus protocols in the aynchronous model with just one failure, Ben-Or and Bracha demonstrated that the problem could be solved with randomness, even for Byzantine failures. Both protocols are natural and intuitive to verify, and Bracha's achieves optimal resilience. However, the expected running time of these protocols is exponential in general. Recently, Kapron, Kempe, King, Saia, and Sanwalani presented the first efficient Byzantine agreement algorithm in the asynchronous, full information model, running in polylogarithmic time. Their algorithm is Monte Carlo and drastically departs from the simple structure of Ben-Or and Bracha's Las Vegas algorithms.

In this paper, we begin an investigation of the question: to what extent is this departure necessary? Might there be a much simpler and intuitive Las Vegas protocol that runs in expected polynomial time? We will show that the exponential running time of Ben-Or and Bracha's algorithms is no mere accident of their specific details, but rather an unavoidable consequence of their general symmetry and round structure. We define a natural class of ``fully symmetric round protocols" for solving Byzantine agreement in an asynchronous setting and show that any such protocol can be forced to run in expected exponential time by an adversary in the full information model. We assume the adversary controls $t$ Byzantine processors for $t = cn$, where $c$ is an arbitrary positive constant $< \frac{1}{3}$. We view our result as a step toward identifying the level of complexity required for a polynomial-time algorithm in this setting, and also as a guide in the search for new efficient algorithms.
\end{abstract}

\section{Introduction} Byzantine agreement is a fundamental problem in distributed computing, first posed by Pease, Shostak, and Lamport \cite{PSL80}. It requires $n$ processors to agree on a bit value despite the presence of failures. We assume that at the outset of the protocol, an adversary has corrupted some $t$ of the $n$ processors and may cause these processors to deviate arbitrarily from the prescribed protocol in a coordinated malicious effort to prevent agreement. Each processor is given a bit as input, and all good (i.e. uncorrupted) processors must reach agreement on a bit which is equal to at least one of their input bits. To fully define the problem, we must specify the model for communication between processors, the computational power of the adversary, and also the information available to the adversary as the protocol executes. We will work in the message passing model, where each pair of processors may communicate by sending messages along channels. It is assumed that the channels are reliable, but asynchronous. This means that a message which is sent is eventually received (unaltered), but arbitrarily long delays are allowed. We assume that the sender of a message is always known to the receiver, so the adversary cannot ``impersonate" uncorrupted processors.

We will be very conservative in placing limitations on the adversary. We consider the \emph{full information} model, which allows a computationally unbounded adversary who has access to the entire content of all messages as soon as they are sent. We allow the adversary to control message scheduling, meaning that message delays and the order in which messages are received may be maliciously chosen. One may consider a \emph{non-adaptive} adversary, who must fix the $t$ faulty processors at the beginning of the protocol, or an \emph{adaptive} adversary, who may choose the $t$ faulty processors as the protocol executes. Since we are proving an impossibility result, we consider non-adaptive adversaries (this makes our result stronger). We will consider values of $t$ which are $= cn$ for some positive constant $c < \frac{1}{3}$. (The problem is impossible to solve if $t \geq \frac{n}{3}$.) We define the running time of an execution in this model to be the maximum length of any chain of messages (ending once all good processors have decided).

In the asynchronous setting, the seminal work of Fischer, Lynch, and Paterson \cite{FLP83} proved that no deterministic algorithm can solve Byzantine agreement, even for the seemingly benign failure of a single unannounced processor death. More specifically, they showed that any deterministic algorithm may fail to terminate. In light of this, it is natural to consider randomized algorithms with a relaxed termination requirement, such as terminating with probability one. In quick succession following the result of \cite{FLP83}, Ben-Or \cite{BO83} and Bracha \cite{B84} each provided randomized algorithms for asynchronous Byzantine agreement terminating with probability one and tolerating up to $t < \frac{n}{5}$ and $t < \frac{n}{3}$ faulty processors respectively. These algorithms feature a relatively simple and intuitive structure, but suffer greatly from inefficiency, as both terminate in expected exponential time. However, when the value of $t$ is very small, namely $\mathcal{O}(\sqrt{n})$, the expected running time is constant.

This state of affairs persisted for a surprising number of years, until the recent work of Kapron, Kempe, King, Saia, and Sanwalani \cite{K08} demonstrated that polynomial-time (in fact, polylogarithmic time) solutions are possible. They presented a polylogarithmic-time algorithm tolerating up to $(\frac{1}{3}-\epsilon) n$ faulty processors (for any positive constant $\epsilon$) which is Monte Carlo and succeeds with probability $1 - o(1)$ \cite{K10}. The protocol is quite technically intricate and has a complex structure. It subtly combines and adapts several core ingredients: Feige's lightest bin protocol \cite{F99}, Bracha's exponential time Byzantine agreement protocol (run by small subsets of processors) \cite{B84}, the layered network structure introduced in \cite{KSSV06a,KSSV06b}, and averaging samplers.

This protocol is a great theoretical achievement, but its use of samplers in particular would pose a challenge to anyone attempting to implement and use the protocol. The authors note: ``For the use of these samplers in our protocols, we assume either a nonuniform model in which each processor has a copy of the required samplers for a given input size, or else that each processor initializes by constructing the required samplers in exponential time. Alternatively, we could use versions of the efficient constructions given in \cite{GVZ06} at the expense of a polylogarithmic overhead in the overall running time of the protocol" \cite{K10}. There is no proof given for this alternative, and there is no further discussion of how this should be implemented. Also, having hard-coded copies of the samplers for a fixed size (or a small number of sizes) stored in the processors may significantly limit flexibility in practice, as one may want to routinely change the number of processors in the system.

Additionally, it seems quite hard to adapt the techniques of Kapron et al. to obtain a Las Vegas algorithm and/or an algorithm against an adaptive adversary, since their protocol relies heavily on universe reduction to ultimately reduce to a very small set of processors. Once we reduce to considering a small subset of the processors, an adaptive adversary could choose to corrupt the entire subset. Even against a non-adaptive adversary,
there is always some chance that the small subset we ultimately choose will contain a high percentage of faulty processors. This is essentially why the Kapron et al. protocol incurs a (small) nonzero probability of failure. We note there are other techniques that may be useful in the Monte Carlo setting but also seem difficult to adapt to the Las Vegas setting. For example, eliminating processors who send messages which are unlikely to have been sent by good processors may be a successful strategy for a Monte Carlo algorithm, but a Las Vegas algorithm cannot risk eliminating many processors for doing things that may be done by good processors with small probability, since it must avoid incurring a nonzero chance of eliminating too many good processors and failing.

Compared to the protocols of Ben-Or \cite{BO83} and Bracha \cite{B84}, the Kapron et al. protocol \cite{K08,K10} appears to be a distant point in what may be a large landscape of possible algorithms. The full range of behaviors and tradeoffs offered by this space remains to be explored. Many interesting questions persist: is there a Las Vegas algorithm that terminates in expected polynomial time? Is there an expected polynomial time algorithm against an adaptive adversary? Is there a much simpler algorithm that performs comparably to the Kapron et al. algorithm, or at least runs in polynomial time with high probability?

In this work, we investigate why simple Las Vegas algorithms in the spirit of \cite{BO83,B84} cannot deliver expected polynomial running time for linear values of $t$ (i.e. $t = cn$ for some positive constant $c$). More precisely, we define a natural class of protocols which we call \emph{fully symmetric round protocols}. This class encompasses Ben-Or \cite{BO83} and Bracha's protocols \cite{B84}, but is considerably more general. Roughly speaking, a protocol belongs to this class if all processors follow the same program proceeding in broadcast rounds where the behavior is invariant under permutations of the identities of the processors attached to the validated messages in each round. In other words, a processor computes its message to broadcast in the next round as a randomized function of the set of messages it has validated, without regard to their senders. We additionally constrain the protocols in the following way. Whenever a processor chooses its message randomly, it must choose from a \emph{constant number} of possibilities. This means that at each step of the protocol, a processor will make a random choice between at most $R$ alternatives, where $R$ is a fixed constant. Note that the set of alternatives itself can vary; it is only the maximum \emph{number} of choices that is fixed. We give a formal description of fully symmetric round protocols in Section \ref{sec:symmetry}. We will prove that for any algorithm in this class which solves asynchronous Byzantine agreement, there exists some input values and some adversarial strategy which causes the expected running time to be exponential in $n$, when $t =cn$ for any fixed positive constant $c$.

Our general proof strategy is to consider a chain of $E$-round executions (for some suitably large value $E$) where the behavior of some good processors is the same between any two adjacent executions in the chain, and the two ends of the chain must have different decision values. This implies that some execution in the chain must not have terminated within $E$ rounds. This is reminiscent of a strategy often used to prove a lower bound of $t$ rounds for deterministic protocols in the synchronous setting (see \cite{DS82} for example). Employing this sort of strategy for randomized algorithms presents an additional challenge, since any particular execution may be very unlikely. To address this, we consider classes of closely related executions where an adversary is able to exert enough control over a real execution to force it to stay within a chosen class with significant probability.

We view this work not as a primarily negative result, but rather as a guide in the search for new efficient Byzantine agreement algorithms in the asynchronous, full information setting. The goal of this paper is to illuminate some of the obstacles that must be surmounted in order to find an efficient Las Vegas protocol and to spur new thinking about protocols which lie outside the confines of our impossibility result without requiring the full complexity of the Kapron et al. protocol. We hope that the final outcome of this line of research will be interesting new algorithms as well as a greater understanding of the possible features and tradeoffs for protocols in this environment.

\subsection{Other Related Work}
Asynchronous Byzantine agreement has also been studied in the setting where cryptographic primitives are available (for this, the adversary must be assumed to be computationally bounded). Both Rabin \cite{R83} and Toueg \cite{T84} presented solutions in this model, supposing that messages are authenticated by digital signatures and processors share a secret sequence of random bits supplied in advance by a trusted dealer. Both solutions terminate in a small constant number of expected rounds. Assuming private channels between pairs of processors, Berman and Garay \cite{BG93} and Canetti and Rabin \cite{CR93} provided additional solutions. Work in the cryptographic setting has ultimately led to protocols that terminate in constant expected time, have optimal resilience ($t < \frac{n}{3}$) and send $\mathcal{O}(n^2)$ messages (protocols provided by Cachin, Kursawe, and Shoup \cite{CKS05} and Nielson \cite{N02}).

In the synchronous, full-information setting, polylogarithmic round randomized protocols for byzantine agreement against a non-adaptive adversary were given by King, Saia, Sanwalani, and Vee \cite{KSSV06a,KSSV06b}, Ben-Or, Pavlov, and Vaikuntanathan \cite{BPV06}, and Goldwasser, Pavlov, and Vaikuntanathan \cite{GPV06}. Restricting the adversary to be non-adaptive is necessary to achieve polylogarithmic time protocols (for values of $t$ which are linear in $n$), since Bar-Joseph and Ben-Or \cite{BB98} have proven that any randomized, synchronous protocol against a fail-stop, full information adversary who can adaptively fail $t$ processors must require at least $\frac{t}{\sqrt{n\log n}}$ rounds in expectation.

Another lower bound for randomized Byzantine agreement protocols was proven by Attiya and Censor \cite{AC08}, who showed that for each integer $k$, the probability that a randomized Byzantine agreement algorithm tolerating $t$ faults with $n$ processors does not terminate in $k (n-t)$ steps is at least $1/c^k$ for some constant $c$. This bound holds even against a considerably weaker adversary than we are considering.

Recent work of King and Saia \cite{KS09,KS10} has provided Byzantine agreement protocols in the synchronous setting with reduced communication overhead, namely $\tilde{\mathcal{O}}(n^{3/2})$ bits in the full information model against a non-adaptive adversary \cite{KS09}, and $\tilde{\mathcal{O}}(\sqrt{n})$ bits against an adaptive adversary under the assumption of private channels between all pairs of processors \cite{KS10}.

The use of averaging samplers in recent protocols is foreshadowed by a synchronous protocol presented by Bracha \cite{B85} that assigned processors to committees in a non-constructive way. Chor and Dwork \cite{CD89} provide an excellent survey that covers this as well as the other early work we have referenced.

\section{Preliminaries}\label{sec:prelim}
We begin by formally specifying the model and developing a needed mathematical definition.

\subsection{The Asynchronous, Full Information Message Passing Model and Randomized Algorithms}
We consider $n$ processors who communicate asynchronously by sending and receiving messages. We assume that the communication channel between two processors never alters any messages, and that the sender of a message can always be correctly determined by the receiver. To model asynchrony, we follow the terminology of \cite{FLP83}. We suppose there is a \emph{message buffer}, which contains all messages which have been sent but not yet received. A \emph{configuration} includes the internal states of all processors as well as the contents of the message buffer. A protocol executes in \emph{steps}, where a single step between two configurations consists of a single processor $p$ receiving a value $m$ from the message buffer, performing local computation (which may involve randomness), and sending a finite set of messages to other processors (these are placed in the message buffer). We note that the value returned by the message buffer is either a message previously sent to $p$ or $\emptyset$ (which means that no message is received). The only constraint on the non-deterministic behavior of the message buffer is that if a single processor $p$ takes infinitely many steps, then every message sent to $p$ in the message buffer is eventually received by $p$.

We suppose there is an \emph{adversary} who controls some $t$ of the processors. We assume these $t$ processors are fixed from the beginning of the protocol. These will be called the \emph{faulty} processors, while the other processors will be called \emph{good} processors. The faulty processors may behave arbitrarily and deviate from the protocol in malicious ways. The adversary also controls the message scheduling (i.e. it decides which processor takes the next step and what the message buffer returns, subject to the constraints mentioned above). Our adversary is computationally unbounded, and has access to the content of all messages as soon as they are sent. Based on this information, the adversary can adaptively decide in what order to deliver messages, subject only to the constraint that all messages which are sent between good processors must eventually be delivered.

We model the use of randomness in a protocol by allowing each processor to sample from its own source of randomness, which is independent of the sources sampled by other processors and \emph{unpredictable} to the adversary. This means that before a good processor samples from its random source, the adversary will know only the distribution of the possible outcomes and nothing more.

We note that the outcome of a step of the protocol taken by a processor $p$ is determined by the configuration before the step, the message (or $\emptyset$) received by $p$ at the beginning of the step, and the local randomness of $p$. For steps where no randomness is used, the outcome is determined by the prior configuration and the received message only. We let $r_p$ denote the local randomness of processor $p$ sampled during a step, and we refer to $e:=(p,m,r_p)$ as an \emph{event}. If $C$ denotes the current configuration, then $e(C)$ denotes the new configuration resulting from this event. If the message $m$ is either $\emptyset$ or is in the message buffer (and intended for $p$) in the configuration $C$, then we say that $e$ can be \emph{applied} to $C$. We define a \emph{schedule} from $C$ to be a sequence of events that can be applied consecutively, beginning with $C$. We note that for steps involving non-empty randomness $r_p$, the adversary does not have full control over the event: it can only choose the message scheduling, and has no control over the local randomness of a good processor $p$. In fact, when the adversary chooses a processor $p$ to take the next step and the message to deliver, it cannot predict the value of $r_p$ that will be sampled when $p$ is a good processor (before $r_p$ is sampled by $p$, the adversary only knows what distribution it will be sampled from).

\subsection{Adjusting Probability Distributions}\label{sec:distribution} We will constrain our fully symmetric round protocols to always choose the next message randomly via some distribution on at most $R$ possibilities, where $R$ is a fixed constant. We note that the possible messages themselves can change according to the state of the processor as the protocol progresses: it is only the \emph{number} of choices that is constrained, not the choices themselves. Since the probability distributions on $R$ values can be arbitrary, we will define closely related distributions which have more convenient properties for our analysis.

We let $\mathcal{D}$ denote a distribution on a set $\mathcal{S}$ of size at most $R$. We let $\rho_s$ denote the probability that $\mathcal{D}$ places on $s \in \mathcal{S}$. In our proof, we will be considering $t$ samples of such a distribution $\mathcal{D}$. For each $s \in \mathcal{S}$, the expected number of times that $s$ occurs when $t$ independent samples of $\mathcal{D}$ are taken is $\rho_s t$. In general, this may not be a integer. We will prefer to work with integral expectations, so we define an alternate distribution $\widetilde{\mathcal{D}}$ on the same set $\mathcal{S}$. We let $\tilde{\rho}_s$ denote the probability that $\widetilde{\mathcal{D}}$ places on $s$ for each $s \in S$. The definition of $\widetilde{\mathcal{D}}$ is motivated by two goals: we will ensure that $\tilde{\rho}_s t$ is an positive integer for each $s \in S$, and also that $\tilde{\rho}_s$ and $\rho_s$ are sufficiently close for each $s \in \mathcal{S}$.

Since the size of $\mathcal{S}$ is at most $R$, there must exist some $s^* \in S$ such that $\rho_{s^*} \geq \frac{1}{R}$. We fix this $s^*$, and we also fix a small real number $\epsilon > 0$ (whose precise size with respect to $t, R$ will be specified later). For all $s \in \mathcal{S} - \{s^*\}$, we define $\tilde{\rho}_s$ to be the least positive integer multiple of $\frac{1}{t}$ which is $\geq \max \{\rho_s, \epsilon\}$. For $s^*$, we define $\tilde{\rho}_{s^*} = 1 - \sum_{s \in \mathcal{S} - \{s^*\}} \tilde{\rho}_s$.

\begin{lemma}\label{lem:distribution} When $t >  R^2$ and $0<\epsilon < \frac{1}{R^2} - \frac{1}{t}$, $\widetilde{\mathcal{D}}$ is a probability distribution on $\mathcal{S}$, and $\tilde{\rho}_s t$ is a positive integer for each $s \in \mathcal{S}$.
\end{lemma}

\begin{proof} By definition of $\tilde{\rho}_{s^*}$, we see that $\sum_{s \in \mathcal{S}} \tilde{\rho}_s = 1$. To show that $\widetilde{\mathcal{D}}$ is a valid distribution, it remains to prove that $0 \leq \tilde{\rho}_s \leq 1$ for every $s \in \mathcal{S}$. For $s \in \mathcal{S} - \{s^*\}$, $\tilde{\rho}_s \geq \epsilon >0$. Also, $\tilde{\rho}_s < \max \{\epsilon, \rho_s\} + \frac{1}{t}$. Since $s \neq s^*$, $\max\{ \epsilon, \rho_s\} \leq 1 - \frac{1}{R}$. Thus, $\tilde{\rho}_s < 1 - \frac{1}{R} + \frac{1}{t}$. Since $t > R$, this quantity is $ < 1$.

For $s^*$, it is clear that $\tilde{\rho}_{s^*} < 1$, since $\sum_{s \in \mathcal{S} - \{s^*\}} \tilde{\rho}_s > 0$. For each $s \in \mathcal{S} - \{s^*\}$, we have $\tilde{\rho}_s - \rho_s < \epsilon + \frac{1}{t}$. Therefore,
\[\sum_{s \in \mathcal{S} - \{s^*\}} \tilde{\rho}_s < \sum_{s \in \mathcal{S} - \{s^*\}} \left(\rho_s + \epsilon + \frac{1}{t}\right).\]
Since the size of $\mathcal{S}$ is at most $R$ and $\sum_{s \in \mathcal{S} - \{s^*\}} \rho_s = 1 - \rho_{s^*}$, we may conclude:
\[\sum_{s \in \mathcal{S} - \{s^*\}} \tilde{\rho}_s < 1 - \rho_{s^*} + R\left(\epsilon + \frac{1}{t}\right).\]
Thus,
\[\tilde{\rho}_{s^*} > \rho_{s^*} - R\left ( \epsilon + \frac{1}{t}\right) >0,\]
since $\epsilon+ \frac{1}{t} < \frac{1}{R^2}$ and $\rho_{s^*} \geq \frac{1}{R}$.
This shows that $\widetilde{\mathcal{D}}$ is indeed a probability distribution on $\mathcal{S}$.

For $s \neq s^*$, $\tilde{\rho}_s t \in \mathbb{Z}$ follows simply from the fact that $\tilde{\rho}_s$ was chosen to be an integral multiple of $\frac{1}{t}$. Since each $\tilde{\rho}_s$ for $s \in \mathcal{S} -\{s^*\}$ is an integral multiple of $\frac{1}{t}$, so is $\tilde{\rho}_{s^*} = 1- \sum_{s \in \mathcal{S} - \{s^*\}} \tilde{\rho}_s$. Hence $\tilde{\rho}_{s^*} t$ is also a positive integer.
\end{proof}

We will additionally use the following consequence of the Chernoff bound. The proof can be found in Appendix \ref{app:chernoff}.
\begin{lemma}\label{lem:probability} Let $\mathcal{D}$ be an arbitrary distribution on a set $\mathcal{S}$ of at most $R$ possible values, and let $\widetilde{\mathcal{D}}$ be defined from $\mathcal{D}$ as above, with $t=cn >  (\frac{2}{c})R^2$ and $\epsilon < \frac{c}{2R^2} - \frac{1}{t}$ (where $c$ is a positive constant satisfying $0 < c < \frac{1}{3}$). Let $s \in \mathcal{S}$, and let $\rho_s, \tilde{\rho}_s$ denote the probabilities that $\mathcal{D}$ and $\widetilde{\mathcal{D}}$ assign to this value, respectively. Let $X_1, \ldots, X_{(1-c)t}$ denote independent random variables, each equal to 1 with probability $\rho_s$ and equal to 0 with probability $1 - \rho_s$. Then:
\[\mathbb{P}\left[ \sum_{i=1}^{(1-c)t} X_i \geq \tilde{\rho_s}t\right] \leq e^{-\delta  c^3 n/(3(1-c))},\]
where $\delta$ is defined to be the minimum of $\epsilon$ and $\frac{1}{4R}$.
\end{lemma}

\section{Fully Symmetric Round Protocols}
\label{sec:symmetry}
We now define the class of fully symmetric round protocols. In these protocols, communication proceeds in rounds. These are similar to the usual notion of rounds in the synchronous setting, but in the asynchronous setting a round may take an arbitrarily long amount of time and different processors may be in different rounds at any given time.
Our definition is motivated by the core structure of Bracha's protocol, so we first review this structure. Bracha's protocol relies on two primitives, called Broadcast and Validate. The broadcast primitive allows a processor to send a value to all other processors and enforces that even faulty processors must send the same value to everyone or no value to anyone. The Validate primitive essentially checks that a value received via broadcast could have been sent by a good processor (we elaborate more on this below). Bracha describes the basic form of a round of his protocol as follows\footnote{Bracha refers to this as a ``step"  \cite{B84} and uses the terminology of ``round" a bit differently.}:

\vspace*{0.5cm}
\textbf{round(k)}

\hspace*{0.5cm} \textit{Broadcast(v)}

\hspace*{0.5cm} wait till \textit{Validate} a set $S$ of $n-t$ $k$-messages

\hspace*{0.5cm} $v:= N(k,S)$
\vspace*{0.5cm}

Here, a $k$-message is a message broadcast by a processor in round $k$, and $N$ is the \emph{protocol function} which determines the next value to be broadcast ($N$ is randomized). In Bracha's protocol, $N$ considers only the set of $k$-messages themselves, and does not consider which processors sent them. This is the ``symmetric" quality which we will require from fully symmetric round protocols. This structure and symmetry also characterize Ben-Or's protocol \cite{BO83}, except that the Broadcast and Validate primitives are replaced just by sending and receiving. We will generalize this structure by allowing protocol functions $N$ which consider messages from earlier rounds as well.

Fully symmetric round protocols will invoke two primitives, again called Broadcast and Validate. We assume these two primitives are instantiated by \emph{deterministic} protocols. In each asynchronous round, a processor invokes the Broadcast primitive and broadcasts a message to all other processors (that message will be stamped with the round number). We will describe the properties of the broadcast primitive formally below. To differentiate from the receiving of messages (which simply refers to the event of a message arriving at a processor via the communication network), we say a processor $p$ \emph{accepts} a message $m$ when $p$ decides that $m$ is the outcome of an instantiation of the broadcast primitive. When we refer to the \emph{round number of a message}, we mean the round number attached to the message by its sender.
For a fully symmetric round protocol, a round can be described as follows:

\vspace*{0.5cm}
\textbf{round(k)}

\hspace*{0.5cm} \textit{Broadcast(v)}

\hspace*{0.5cm} wait till \textit{Validate} a set $S$ of $n-t$ $k$-messages

\hspace*{0.5cm} let $S'$ denote the set of all validated $i$-messages for all $i < k$

\hspace*{0.5cm} $v := N(k, S\cup S')$
\vspace*{0.5cm}

The message to be broadcast in the first round is computed as $N(0,b)$, where $b$ is the input bit of the processor. As in the case of Bracha's protocol, we consider the set of messages $S \cup S'$ as divorced from the sender identities, so the protocol function $N$ does not consider which processor sent which message. Note here that we have allowed the protocol function to consider all currently validated messages with round numbers $\leq k$ (i.e. were broadcast by their senders in rounds $\leq k$). In contrast, the Validate algorithm may consider the processor identities attached to messages.

In summary, in each round a processor waits to validate $n-t$ messages from other processors for that round. Once this occurs, it applies the protocol function $N$ to the set of validated messages. This protocol function determines whether or not the processor decides on a final bit value at this point (we assume this choice is made deterministically), and also determines the message to be broadcast in the next round. This choice may be made randomly. We note that choice of whether to decide a final bit value (and what that value is) only depends on the set of accepted messages themselves, and does not refer to the senders.

\paragraph{Key Constraint on Randomized Behavior} We constrain a processor's random choices in the following crucial way. We assume that when a processor employs randomness to choose its message to broadcast in the next round, it chooses from at most $R$ possibilities, where $R$ is a fixed, global constant independent of all other parameters (e.g. it does not depend on the round number or the total number of processors)\footnote{This constraint is satisfied by Ben-Or and Bracha's protocols, since both choose from two values whenever they choose randomly.}. Note that the choices themselves may depend on the round number, the total number of processors, etc. The messages themselves may also be quite long - there is no constraint on their bit length.

\paragraph{Full Symmetry} Fully symmetric round protocols are invariant under permutations of the identities associated with validated messages in each round. At the end of each round, a good processor may consult all previously validated messages (divorced from any information about their senders) and must choose a new message to broadcast at the beginning of the next round. It may make this choice randomly, so we think of the set $S \cup S'$ of all previously validated messages as determining a distribution on a \emph{constant} number of possible messages for the next round. We emphasize that since $S \cup S'$ is just the set of the bare messages themselves, it also contains no information about which messages were sent by the same processors, so the distribution determined by $S \cup S'$ is invariant under all permutations of the processor identities associated with messages for each round where the permutations may \emph{differ} per round.


\paragraph{Broadcast and Validate Primitives} We now formally define the properties we will assume for the broadcast and validate primitives. We recall that these are assumed to be deterministic. We first consider broadcast. We suppose that the broadcast primitive is invoked by a processor $p$ in order to send message $m$ to all other processors. We consider the $n$ processors as being numbered 1 through $n$, which allows us to identify the set of processors with the set $[n]:= \{1, 2, \ldots, n\}$. We will assume that for each permutation $\pi$ of the set of $[n]$, there exists a finite schedule of events that can be applied (starting from the current configuration) such that at the end of the sequence of events, all processors have accepted the message $m$ and that for each $i$ from 1 to $n-1$, there is a prefix of the schedule such that at the end of the prefix, exactly the processors $\pi(1), \ldots, \pi(i)$ have accepted the message $m$. Essentially, this means that every possible order of acceptances can be achieved by some applicable schedule. (Note that within these schedules, all processors act according to the protocol.) More formally, we make the following definition:

\begin{definition} We say a broadcast protocol allows \textbf{arbitrary receiver order} if for any processor $p$ invoking the protocol to broadcast a message $m$ and for any permutation $\pi$ of $[n]$, there exists a finite schedule $\sigma_{\pi}$ of events that can be applied consecutively starting from the initial configuration such that there exist prefixes $\sigma_1, \ldots, \sigma_n = \sigma_{\pi}$ of $\sigma_{\pi}$ such that in the configuration resulting from $\sigma_i$, exactly the processors $\pi(1), \ldots, \pi(i)$ have accepted $m$, and no other processors have.
\end{definition}

It is clear that this property holds if one implements broadcast simply by invoking the send and receive operations on the communication network. This property also holds for Bracha's broadcast primitive, which enforces that even faulty processors must send the same message to all good processors or no message at all. This property will be useful to our adversary (who controls scheduling) because it allows complete control over the order in which processors accept messages. We assume that our fully symmetric round protocols treat each invocation of the broadcast primitive as ``separate" from the rest of the protocol in the sense that any messages sent not belonging to an instance of the broadcast primitive do not affect a processor's behavior within this instance of the broadcast primitive.

We consider the Validate primitive as an algorithm $V$ which takes as input the set of all accepted messages so far along with accompanying information specifying the sender of each message. The algorithm then deterministically proceeds to mark some subset of the previously accepted messages as ``validated". We assume that this algorithm is monotone in the following sense. We let $W^+ \subseteq S^+$ be two sets of accepted message and sender identity pairs (we use the $^+$ symbol to differentiate these sets of messages with senders from sets of messages without sender identity attached). Then if a message, sender pair $(m,p) \in W^+$ is marked valid by $V(W^+)$, then this same pair $(m,p)$ will be marked valid by $V(S^+)$ as well. In other words, marking a message as valid is a decision that cannot be reversed after new messages are accepted.

We assume the validation algorithm is called each time a new message is accepted to check if any new messages can now be validated. Bracha's Validate algorithm is designed to validate only messages that could have been sent by good processors in each round. It operates by validating an accepted message $m$ for round $k$ if and only if there are $n-t$ validated messages for round $k-1$ that could have caused a good processor to send $m$ in round $k$ (i.e. $m$ is an output of the protocol function $N$ that occurs with nonzero probability when these $n-t$ validated round $k-1$ messages are used as the input set). In the context of Bracha's algorithm, where the behavior for one round only depends on the messages from the previous round, this essentially requires faulty processors to ``conform with the underlying protocol" \cite{B84} (up to choosing their supposedly random values maliciously) or have their messages be ignored.

In the context of protocols that potentially consider messages from \emph{all} previous rounds, one might use a stronger standard for validation. For instance, to validate a message $m_k$ for round $k$ sent by a processor $p$, one might require that there are messages $m_1, \ldots, m_{k-1}$ for rounds 1 through $k-1$ sent by $p$ which are validated and that there are sets of validated messages $S_1, \ldots, S_{k-1}$ such that each $S_i$ contains messages for rounds $\leq i$ and exactly $n-t$ messages for round $i$, $S_i \subset S_{i+1}$ for each $i < k-1$, and $N(i, S_i) = m_{i+1}$ with non-zero probability for each $i$ from 1 to $k-1$. This essentially checks that there is a sequence of sets of validated messages that $p$ could have considered in each previous round that would have caused a good processor to output messages $m_1, \ldots, m_k$ in rounds 1 through $k$ with non-zero probability.

Roughly, we will allow all validation algorithms that never fail to validate a message $m$ sent by a good processor $p$ when all of the previous messages sent by $p$ and all of the messages that caused the processor $p$ to send the message $m$ have been accepted. We call such validation algorithms \emph{good message complete}. We define this formally as follows.

\begin{definition} A Validate algorithm $V$ is \textbf{good message complete} if the following condition always holds.
Suppose that $S^+_1 \subset S^+_2 \subset \ldots  \subset S^+_k$ are sets of validated messages (with sender identities attached) such that each $S^+_i$ contains exactly $n-t$ round $i$ messages and $m_1, \ldots, m_k$ occur with non-zero probability as outcomes of $N(1,S_1), \ldots, N(k,S_k)$ respectively. Then if a set $W^+$ of messages (with sender identities attached) includes $S^+_k$ as well as the messages $m_1, \ldots, m_k$ from the same sender, then $V(W^+)$ marks $m_k$ as validated.
\end{definition}

This means that if during a real execution of the protocol, a good processor $p$ computes its first $k$ messages $m_1, \ldots, m_k$ by applying $N(1, S_1), N(2, S_2), \ldots, N(k, S_k)$ respectively and another good processor $q$ has accepted $m_1, \ldots, m_k$ from $p$ as well as all of the messages in $S_k$, then $q$ will validate $m_k$.

We note that the use of Validate protocols which are \emph{not} good message complete seems quite plausible in the Monte Carlo setting at least, since a Monte Carlo algorithm can afford to take some small chance of not validating a message sent by a correct processor. In this case, it would also be plausible to consider randomized validation protocols. However, since we are considering Las Vegas algorithms, we will restrict our attention in this paper to deterministic, good message complete validation protocols.

We have now completed our description of fully symmetric round protocols. In summary, they are round protocols that invoke a broadcast primitive allowing arbitrary receiver order, invoke a Validate primitive that is good message complete, are invariant under permutations of the processor identities attached to the messages in each round, and always make random choices from a constant number of possibilities.

\section{Impossibility of Polynomial Time for Fully Symmetric Round Protocols}
We are now ready to state and prove our result.

\begin{theorem}\label{thm:main} For any fully symmetric round protocol solving asynchronous Byzantine agreement with $n$ processors for up to $t = cn$ faults, for $c>0$ a positive constant, there exist some values of the input bits and some adversarial strategy resulting in expected running time that is exponential in $n$.
\end{theorem}

\begin{proof} We suppose we have a fully symmetric round protocol with resilience $t = cn$. We will assume that $t$ divides $n$ for convenience, and also that $ct$ is an integer. These assumptions will make our analysis a little cleaner, but could easily be removed.
We let $R$ denote the constant bound on the number of possibilities for each random choice. $E$ will denote a positive integer, the value of which will be specified later. (It will be chosen as a suitable function of $c, R$ and $n$, and will be exponential in $n$ when $c$, $R$ are positive constants.)

We will be considering partial executions of the protocol lasting for $E$ rounds. For convenience, we think of our protocol as continuing for $E$ rounds even if all good processors have already decided (this can be artificially achieved by having decided processors send default messages instead of terminating). The adversary must fix $t$ faulty processors at the beginning of an execution. Once these processors are fixed, we divide the $n$ processors into disjoint groups of $t$ processors each, so there are $\frac{n}{t}$ groups. We will refer to the groups as $G_1$ through $G_{n/t}$. We choose our groups so that exactly $ct$ of the processors in each group are faulty. The main idea of our proof is as follows. Since the broadcast and validation primitives essentially constrain the behavior of faulty processors, we think of the adversary as controlling only the (supposedly) random choices of faulty processors as well as the message scheduling. This means that when faulty players invoke the randomized function $N(i,S)$, they may maliciously chose any output that occurs with nonzero probability. In all other respects, they will follow the protocol.

The adversary will choose the message scheduling so that the $t$ processors in a single group will proceed in lockstep: the sets of messages that they use as input to $N$ will always be the same in each round. This means that all $t$ processors in a group will be choosing their next message from the \emph{same distribution}. Since there are only a constant number of possibilities and the adversary controls a constant fraction of the processors in the group, it can ensure with high probability that the collection of messages which are actually chosen is precisely equal to the expectation under the adjusted distribution. More precisely, we let $\mathcal{D}$ denote the distribution (on possible next round messages) resulting from applying $N$ to a particular set of messages in a particular round, and we let $\mathcal{S}$ denote the set of (at most $R$) outputs that occur with nonzero probability. We define $\widetilde{D}$ with respect to $\mathcal{D}$ as in Section \ref{sec:distribution}. Then, with high probability, once the adversary sees the outputs chosen by the $(1-c)t$ good processors in the group, it can choose the messages of the $ct$ faulty processors in the group so that the total number of processors in the group choosing each $s \in \mathcal{S}$ is exactly $\tilde{\rho}_s t$. (This is proven via Lemma \ref{lem:probability}.)

We can then consider classes of executions which proceed with these groups in lockstep as being defined by the set of messages used as input to $N$ in each round by each group (as well as the sets of sender, round number pairs for the messages, but these pairs are divorced from the messages themselves). With reasonable probability, an adversary who controls the message scheduling and the random choices of faulty processors can force a real execution to stay within such a class for $E$ rounds. We will prove there exists such a class in which some good processors fail to decide in the first $E$ rounds. Putting this all together, we will conclude that there exist some values of the input bits and some adversarial strategy that will result in expected running time that is exponential in $n$.

Our formal proof begins with the following definition.

\begin{definition}\label{def:class} An $E$-round lockstep execution class $\mathcal{C}$ is defined by a setting of the input bits for each group (processors in the same group will have the same input bit), message sets $S_i^j$ for  all $1 \leq i \leq E$ and $1 \leq j \leq n/t$ where $S_i^j$ is used as the input to $N$ in round $i$ by each processor in group $G_j$ during some real execution, and sets $Z_i^j$ of processor, round number pairs consisting of all pairs $(p,k)$ such that the message broadcast by processor $p$ in round $k$ is contained in $S_i^j$. We require that for all $i,j$, if $(p,k) \in Z_i^j$ and processors $p$ and $p'$ are in the same group $G_{j'}$, then $(p',k) \in Z_i^j$ as well.
\end{definition}

We have required that an $E$-round lockstep execution class $\mathcal{C}$ describe \emph{some real execution}, but note that such an execution is not unique. There are many possible executions that correspond to the same class $\mathcal{C}$. It is crucial to note that the messages in sets $S_i^j$ are not linked to their sender, round number pairs in $Z_i^j$. In other words, given the sets $Z_i^j$ and $S_i^j$, one has cumulative information about the messages and the senders, but there is no specification of who sent what.

We will construct a chain $\mathcal{C}_0, \mathcal{C}_1, \ldots, \mathcal{C}_L$ of $E$-round lockstep execution classes with the following properties:
\begin{enumerate}
    \item For each $S_i^j$ in each $\mathcal{C}_\ell$, the number of occurrences of each message $s$ is exactly equal to $\tilde{\rho}_s t$, where $\tilde{\rho}_s$ denotes the probability on $s$ in the distribution $\tilde{\mathcal{D}}$ defined from $\mathcal{D}$ as in Section \ref{sec:distribution}, where $\mathcal{D}$ is the distribution on possible messages induced by $N(i-1, S_{i-1}^j)$.
    \item Each $\mathcal{C}_\ell$ and $\mathcal{C}_{\ell+1}$ differ only in the sets $S_i^j$ for \emph{one} group $G_j$.
    \item It is impossible for any good processor to decide the value 1 during an execution in class $\mathcal{C}_0$.
    \item It is impossible for any good processor to decide the value 0 during an execution in class $\mathcal{C}_L$.
\end{enumerate}

Once we have such a chain of $E$-round lockstep execution classes, we may argue as follows. Since each $\mathcal{C}_\ell$ and $\mathcal{C}_{\ell+1}$ differ only in the behavior of processors in a single group, it is impossible for all good processors to have decided 0 in an execution in class $\mathcal{C}_\ell$ and all good processors to have decided 1 in an execution in class $\mathcal{C}_{\ell+1}$. Since the only decision value possible in $\mathcal{C}_0$ is 0 and the only decision value possible in $\mathcal{C}_L$ is 1, there must be some $\mathcal{C}_{\ell^*}$ which leaves some good processors undecided. In other words, any execution in this class $C_{\ell^*}$ does not terminate in $\leq E$ rounds.

Finally, we will show that when the input bits match the inputs for $\mathcal{C}_{\ell^*}$, the adversary can (with some reasonable probability) cause a real execution to fall in class $\mathcal{C}_{\ell^*}$. Since $E$ is exponential in $n$ whenever $R, c$ are positive constants, this will prove that the expected running time in this case is exponential.

We now present a recursive algorithm which generates a chain of $E$-round lockstep execution classes with the properties required above.

\subsection{Generating the Chain of Execution Classes}
We first describe a method for generating an $E$-round lockstep execution class from a setting of the input bits (these must be the same for all processor within a group) and a family of sets $z_i^j\subset [n]$, each of size $n-t$, where $i$ ranges from 1 to $E$ and $j$ ranges from 1 to $\frac{n}{t}$. Each set $z_i^j$ will be the complement of some group $G_{j'}$. This means that $z_i^j$ is a union of all but one of the groups, so each group of $t$ processors is either contained or $z_i^j$ or disjoint from it. Now, we will create an execution where for each round $i$ and each group $G_j$, the $n-t$ round $i$ messages validated by the members of group $G_j$ in round $i$ are precisely those sent by the processors in the set $z_i^j$. This will correspond to an $E$-round lockstep execution class with sets $Z_i^j$ defined as follows.

We define the sets $Z_i^j$ inductively. Each set $Z_i^j$ consists of pairs $(p,k)$, where $p \in [n]$ is a processor and $k \leq i$ is a round number. For $i=1$, we simply define $Z_1^j$ to consist of the pairs $(p,1)$ where $p \in z_1^j$. Once we have defined $Z_{i-1}^j$, we define $Z_i^j$ to be $Z_{i-1}^j$ plus the sender and round number pairs for any additional messages which may be needed to validate the round $i$ messages of the processors in set $z_i^j$. More formally,
\begin{equation}\label{setbuild}
Z_i^j := Z_{i-1}^j\; \bigcup \;\{(p,k) | p \in z_i^j, k\leq i\} \bigcup_{j' s.t. G_{j'} \cap z_i^j \neq \emptyset} Z_{i-1}^{j'}.
\end{equation}
The second set in this union corresponds to the set of messages sent by processors in $z_i^j$ for all rounds $\leq i$, while the final set in the union contains all the sets $Z_{i-1}^{j'}$ for groups $G_{j'}$ that intersect $z_i^j$. By good message completeness of our validation algorithm, this suffices to ensure that the round $i$ messages from senders in $z_i^j$ will be validated once all of the messages with sender, round number pairs in $Z_i^j$ have been accepted (in fact all of the messages whose sender, round number pairs appear in $Z_i^j$ will be validated, assuming all processors follow the protocol except for manipulating their supposedly random coins). We note that these sets $Z_i^j$ satisfy the required property that if $(p,k) \in Z_i^j$, $(p',k) \in Z_i^j$ as well as for any $p'$ in the same group as $p$. (This follows from induction on $i$ and the fact that this is holds for the sets $z_i^j$.)

We next use the sets $Z_i^j$ to define a set of permutations (one for each processor, round number pair) that we will use to specify the message scheduling during an execution.

\begin{definition}\label{def:permutations} A set of permutations $\{\pi_{p,i}\}$ on $[n]$ (one for each processor $p$ and each round $1 \leq i \leq E$) corresponding to sets $\{Z_i^j\}$ is defined as follows.
For each processor $p$ and round number $i$, we consider each group index $j$ (from $1$ to $n/t$). For each group $G_j$, there is a minimal round $k$ for which $(p,i) \in Z_k^j$ (if $(p,i)$ is not in any of these sets, then define this minimal $k$ to be $\infty$). This induces an ordering on the groups. We will define $\pi_{p,i}$ by setting $\pi_{p,i}(1), \ldots, \pi_{p,i}(t)$ to be processors in the group with the lowest associated minimal round $k$, setting $\pi_{p,i}(t+1), \ldots, \pi_{p,i}(2t)$ to be processors in the group with the second lowest associated minimal round $k$, and so. (Within a group, the processors can be ordered arbitrarily. If two groups have the same $k$, their order can also be chosen arbitrarily.)
\end{definition}

These permutations $\pi_{p,i}$ will be used to specify the order in which the processors will accept the message broadcast by processor $p$ in round $i$. Note that for any set of permutations $\{\pi_{p,i}\}$, we can choose our message scheduling to achieve them, since our broadcast protocol allows arbitrary receiver order and separate invocations of it in our fully symmetric round protocols are treated independently of each other.

We now consider running an execution as follows. In the first round, each processor in group $G_j$ will choose its message to broadcast from the same distribution $\mathcal{D}_{init}$ over $\leq R$ possibilities induced by $N(0,b_j)$, where $b_j$ is the input bit of processors in group $G_j$
(we are assuming that the input bit is the same for all processors within a group, though different groups may have different inputs). We let $\widetilde{\mathcal{D}}_{init}$ be defined from $\mathcal{D}_{init}$ as in Section \ref{sec:distribution}. We choose a multi-set of $t$ messages from the $\leq R$ possible messages such that the number of occurrences of each is equal to $t$ times its probability under $\widetilde{\mathcal{D}}_{init}$. (This is always possible because $t$ times each probability in $\widetilde{\mathcal{D}}_{init}$ is a positive integer, by design.) We assign these messages arbitrarily to the processors in group $G_j$. (This can be thought of as assigning the random coins of each processor to make these outcomes happen.) For each $p$, we then run part of the finite schedule $\sigma_{\pi_{p,1}}$ for $p$'s broadcast in round 1, stopping when exactly those processors in groups $j'$ such that $(p,1) \in Z_1^{j'}$ have accepted the message. At this point, all processors in each group $G_j$ are ready to compute their round 2 messages by calling $N(1,S_1^j)$, where $S_1^j$ contains the messages whose sender, round number pairs appear in $Z_1^j$. We can now describe the rest of the execution by specifying what happens for an arbitrary round $i$, assuming the previous round is just completing.

We assume (by induction) by that at the end of round $i-1$ when the message to be broadcast in round $i$ is computed, all processors in each group $G_j$ have accepted and validated exactly the set of messages $S_{i-1}^j$ corresponding to the sender, round number pairs in $Z_{i-1}^j$. Hence, every processor in group $G_j$ will be computing its round $i$ message by applying $N(i-1,S_{i-1}^j)$ for the same set $S_{i-1}^j$ of accepted messages. We let $\mathcal{D}_{i-1,j}$ denote the probability distribution on the $ \leq R$ possible outputs of $N(i-1,S_{i-1}^j)$, and we let $\widetilde{\mathcal{D}}_{i,j}$ be defined from $\mathcal{D}_{i,j}$ as in Section \ref{sec:distribution}. We choose a multi-set of $t$ messages from the possible outcomes of $N(i-1,S_{i-1}^j)$ such that the number of occurrences of each is equal to $t$ times its probability under $\widetilde{\mathcal{D}}_{i-1,j}$. We assign these messages arbitrarily to the members of group $G_j$.

Now, each of the processors invokes the broadcast protocol on its round $i$ message. We begin running each corresponding finite schedule $\sigma_{\pi_{p,i}}$, stopping at the point where exactly those processors in groups $j$ such that $(p,i) \in Z_i^{j}$ have accepted the message. Also, for every processor $p$ and every round $k < i$, we continue running the finite schedules $\sigma_{\pi_{p,k}}$ to the point where exactly all processors in groups $j$ such that $(p,k) \in Z_i^j$ have accepted the round $k$ message from $p$. This ensures that every processor in each group $G_j$ will accept and validate precisely the set of messages $S_i^j$ whose sender, round number pairs appear in $Z_i^j$. We continue in this way through $E$ rounds.
This is a real execution that corresponds to an $E$-round execution class with sets $Z_i^j, S_i^j$.

To generate our chain of executions, we begin by defining the initial $\mathcal{C}_0$. This is done by setting all of the input bits equal to 0, and choosing sets $z_i^j \subset [n]$ arbitrarily (each is the complement of a single group). The sets $Z_i^j, S_i^j$ are then derived as above. This gives us an $E$-round lockstep execution class $C_0$ corresponding to some real execution in which all of the input bits are 0. (This ensures property 3 for our chain).

To produce the rest of the chain $C_1, \ldots, C_L$, we employ a recursive algorithm called \textbf{ChainGenerator}.
The algorithm is designed to produce a gradual shift from a lockstep execution class with all input bits equal to 0 to a lockstep execution class with all input bits equal to 1. This is accomplished by changing the inputs of one group at a time. In order to change a single group's inputs without affecting the behavior of other processors through the first $E$ rounds, we must first move to a lockstep execution class where the messages sent by this group are not accepted by processors in other groups until \emph{after} they have completed $E$ rounds. We choose the group size to be $t$ in order to make this possible. (Group sizes $< t$ could also be employed, but once the group size gets too small, the adversary will not have enough control to make the group's cumulative behavior match its adjusted expectation.)

To reach an $E$-round lockstep execution class where a particular group's messages are not heard by other processors, we follow an inductive strategy with round $E$ acting as the base case. Suppose that we want to change the inputs for processors in group $G_j$. We cannot do this immediately if it might affect the behavior of processors outside this group in the first $E$ rounds. We define $i-1$ to be the \emph{earliest} round in which the set $z_{i-1}^{j'}$ for some other group $G_{j'}$ includes a sender in $G_j$. We now seek to change the set $z_{i-1}^{j'}$ to be the complement of group $G_{j}$.
Now we have a new instance of the same problem: in order to change what messages group $G_{j'}$ members accept in round $i-1$ without affecting processors outside of this group, we must first get to a lockstep execution class where processors outside of group $G_{j'}$ do not accept messages sent from group $G_{j'}$ with round numbers $\geq i$ until they have completed $E$ rounds. The important thing to notice here is that the new instance of the problem always involves a \emph{higher} round number. Hence, we can formulate this as a recursion, and eventually we reach a point where it is enough to ensure that the messages of some group $G_{j''}$ with round numbers $\geq i'$ are not heard by some other group $G_{j'''}$ in round $E$. This is now easy to do, since we can arrange for the $n-t$ other round $E$ messages to be validated while we delay the messages with round numbers $\geq i'$ from group $G_{j''}$ to $G_{j'''}$, so the processors in group $G_{j'''}$ can exit round $E$. (Notice here that $E$ will be the earliest round in which any group may receive messages with round number $\geq i'$ from $G_{j''}$, and this ensures that these messages cannot be needed to validate the round $E$ messages of processors outside $G_{j''}$.)

More concretely, to change the inputs of some group $G_{g_1}$ from 0 to 1, we begin by initializing a list of group number, round number pairs with the element $(g_1, 1)$. Having a pair $(g_\ell, r_\ell)$ as the last element of our list means that our goal is to arrive at sets $z_i^j$ such that for all $j\neq g_\ell$ and all $i \geq r_\ell$, $z_i^j$ is the complement of group $G_{g_\ell}$. (In other words, the messages that group $G_{g_\ell}$ sends in rounds $\geq r_\ell$ are not heard by processors in other groups.) If our sets $z_i^j$ do not currently satisfy this, we add a pair $(g_{\ell+1}, r_{\ell+1})$ to the list where $r_{\ell+1}-1$ is the minimal value of $i\geq r_\ell$ such that some $z_i^j$ with $j \neq g_\ell$ includes group $G_{g_\ell}$ (and $g_{\ell+1}$ is the corresponding $j$ value for such a set $z_i^j$). Now our (sub)goal is to arrive at sets $z_i^j$ where the messages sent by group $G_{g_{\ell+1}}$ in rounds $\geq r_{\ell+1}$ are not heard by processors in other groups. Once this holds, we can change the set $z_{r_{\ell+1}-1}^{g_{\ell+1}}$ to be the complement of group number $g_{\ell}$, and we can remove the pair $(g_{\ell+1}, r_{\ell+1})$ from the list. Now, there may be other sets $z_i^j$ with $i \geq r_\ell$ and $j \neq g_\ell$ including group $G_{g_\ell}$ that we will need to deal with next. However, since we always consider the \emph{minimal} such $i$, we will not undue the progress we have made by changing $z_{r_{\ell+1}-1}^{g_{\ell+1}}$ in the process of addressing these other sets. Since there are a finite number of groups and we are considering a finite number of rounds, this process will always eventually terminate.

The full description of the recursive function and the proof that it produces a suitable chain of $E$-round lockstep execution classes is below. The function takes in three arguments: a specification of $n$ input bits (denoted $x_1, \ldots, x_n$), a family of sets $\{z_i^j\}$, and an ordered list $\mathcal{L}$ of pairs: each pair contains a group number and a round number between 1 and $E+1$. We denote the $k^{th}$ element of the list by $(g_k, r_k)$, where $g_k$ is the group number, and $r_k$ is the round number. The round numbers in the list will always be strictly increasing. The first element of the list will always be of the form $(g_1, 1)$. We denote the size of the list by $|\mathcal{L}|$. The input bits $x_1, \ldots, x_n$ will always be consistent within each group.

There is also a required relationship between the ordered list and the sets $z_i^j$. For each pair $(g_{k-1}, r_{k-1})$ on the list that is followed by a pair $(g_k, r_k)$, round $r_{k}-1$ must be the earliest round $\geq r_{k-1}$ in which any set $z_i^j$ for a group $G_j \neq G_{g_{k-1}}$ includes the group $G_{g_{k-1}}$. In other words, for all $r_{k-1} \leq i < r_k -1 $ and all $j \neq g_{k-1}$, $z_i^j$ is is the complement of group $G_{g_{k-1}}$.
This relationship will be maintained in the arguments to the recursive calls the function makes to itself.

We initially call our recursive function \textbf{ChainGenerator} with input bits all equal to 0, the sets $z_i^j$ used in defining $\mathcal{C}_0$, and the list initialized to $(1,1)$. The function then proceeds to call itself with new arguments.
Each time a change is made to the input bits and/or to the sets $z_i^j$, we produce a new execution class generated as above from the new input bits and sets.

\subsubsection{The Recursive Algorithm}

\paragraph{ChainGenerator($(x_1, \ldots, x_n), \{z_i^j\}, \mathcal{L}$)} We set $\ell = |\mathcal{L}|$. If $\ell \geq 2$, we proceed as follows.
We examine the last element of the list $\mathcal{L}$, denoted by $(g_\ell, r_\ell)$. We consider two possible cases. Case 1 occurs when none of the sets $z_i^j$  for values of $j \neq g_\ell$ and $r_\ell \leq i \leq E$ include group $g_\ell$. Case 2 occurs when there is some $z_i^j$ for $j \neq g_\ell$, $r_\ell \leq i \leq E$ that does contain group $g_\ell$.

We first consider case 1.
We define new sets $\tilde{z}_i^j$ as follows. For all $j \neq g_\ell$, we set $\tilde{z}_i^j = z_i^j$ for all $i$ (these sets are unchanged). For $j= g_\ell$, we set $\tilde{z}_i^{g_\ell} = z_i^j$ for all $i \neq r_\ell-1$. We define $z_{r_\ell-1}^{g_\ell}$ to be the complement of group $g_{\ell-1}$.
These new sets are then used (as above) to derive sets $\tilde{Z}_i^j, \tilde{S}_i^j$ corresponding to an $E$-round execution class $\tilde{\mathcal{C}}$, using the (unchanged) input bits $(x_1, \ldots, x_n)$.
We output $\tilde{\mathcal{C}}$ as the next $E$-round lockstep execution class in the chain. We remove $(g_\ell, r_\ell)$ from the list $\mathcal{L}$ to form a new list $\tilde{\mathcal{L}}$. We then call \textbf{ChainGenerator}($(x_1, \ldots, x_n), \{\tilde{z}_i^j\}, \tilde{\mathcal{L}}$).

We observe the following. We know that round number $r_\ell -1$ is the minimum of all round numbers $i\geq r_{\ell-1}$ such that some $z_i^j$ for $j \neq g_{\ell-1}$ includes group $g_{\ell-1}$. (This follows from the required relationship between the ordered list $\mathcal{L}$ and the sets $z_i^j$.) Thus, when we create the new sets $\tilde{Z}_i^j$ from the new $\tilde{z}_i^j$'s, the set $\tilde{Z}_{r_\ell-1}^{g_\ell}$ will no longer include any sender, round numbers with senders in group $g_{\ell-1}$ and round numbers $\geq r_{\ell-1}$. This ensures that the members of group $g_\ell$ can now proceed through round $r_\ell-1$ without accepting any messages from group $G_{g_{\ell-1}}$ with round numbers $\geq r_{\ell-1}$.

To confirm that our constraints on the input arguments are satisfied, note that $\tilde{\mathcal{L}}$ is a sublist of $\mathcal{L}$, so its round numbers remain strictly increasing. Also, we have only changed the sets $z_i^j$ in rounds $i \geq r_{\ell}-1$, so if $\tilde{\mathcal{L}}$ still has size at least two\footnote{Note that this constraint is vacuous when $|\tilde{\mathcal{L}}| = 1$.}, we have preserved the fact that round $r_{\ell-1}-1$ is the earliest round in which any set $\tilde{z}_i^j$ for $j\neq g_{\ell-2}$ and $i \geq r_{\ell-2}$ includes group number $g_{\ell-2}$. (Note that $r_{\ell-1}-1 < r_\ell -1$.)

We now consider case 2. In this case, there is some $j \neq g_\ell$, $r_\ell \leq i \leq E$ such that $z_i^j$ \emph{does} contain group $g_\ell$. Among these $i,j$ values, we fix a pair $(i^*, j^*)$ where $i^*$ is minimal. We define $r_{\ell+1} = i^* + 1$. We note that $r_{\ell+1} > r_\ell$. We define $g_{\ell+1}=j^*$. We append the pair $(g_{\ell+1}, r_{\ell+1})$ to the list $\mathcal{L}$ to form the new list $\mathcal{L}'$. We then call \textbf{ChainGenerator($(x_1, \ldots, x_n),\{z_i^j\}, \mathcal{L'}$)}. Note that in this case, the input bits and the sets $\{z_i^j\}$ are unchanged, so $\{z_i^j\}, \mathcal{L'}$ still satisfy our requirements by construction.

We are left to handle the case of $\ell = 1$. In this case, we have a single pair $(g_1, 1)$ in the list. We again consider two cases. In case 1, none of the sets $z_i^j$ for $j \neq g_1$ and $1 \leq i \leq E$ include the group $g_1$. In case 2, there is some $z_i^j$ for $j\neq g_1$ that does include a pair with sender in group $g_1$.

We consider case 1. We first change the input bits $x_1, \ldots, x_n$ to new bits $\tilde{x}_1, \ldots, \tilde{x}_n$ by setting all of the input bits for processors in group $g_1$ to be 1 (the other inputs remain unchanged). We leave the sets $\{z_i^j\}$ unchanged. Using the input bits $\tilde{x}_1, \ldots, \tilde{x}_n$ and the sets $z_i^j$, we derive sets $Z_i^j, \tilde{S}_i^j$ corresponding to an $E$-round lockstep execution class $\tilde{\mathcal{C}}$. (Note that the sets $Z_i^j$ are unchanged, because they only depend on the $z_i^j$'s and not on the input bits.) We output $\tilde{\mathcal{C}}$ as the next $E$-round lockstep execution class in the chain.
If $g_1 = n/t$, we terminate. Otherwise, we define the new list $\tilde{\mathcal{L}}$ to be $\{(g_1 + 1, 1)\}$, and we call \textbf{ChainGenerator($(\tilde{x}_1, \ldots, \tilde{x}_n), \{z_i^j\}, \tilde{\mathcal{L}}$)}.

We now consider case 2. In this case, there is some $j \neq g_1$, $1 \leq i \leq E$ such that $z_i^j$ \emph{does} contain group $g_1$. Among these $i,j$ values, we fix a pair $(i^*, j^*)$ where $i^*$ is minimal. We define $r_{2} = i^* + 1$. We note that $r_{2} > r_1 =1$. We define $g_{2}=j^*$. We append the pair $(g_{2}, r_{2})$ to the list $\mathcal{L}$ to form the new list $\mathcal{L}'$. We then call \textbf{ChainGenerator($(x_1, \ldots, x_n),\{z_i^j\}, \mathcal{L'}$)}. Note that in this case, everything except the list is unchanged, and $\mathcal{L'}$ satisfies our requirements by construction.
This concludes the description of the algorithm.

\subsubsection{Proof of Correctness for the Algorithm} We now prove that this algorithm produces a chain of $E$-round lockstep execution classes with the desired properties 1 through 4.

\begin{lemma}\label{lem:generator} When called with the initial arguments $(0,0,\ldots,0)$, the sets $\{z_i^j\}$ for $\mathcal{C}_0$, and $\mathcal{L} := \{(1, 1)\}$, the function \textbf{ChainGenerator} eventually terminates and produces a chain of $E$-round lockstep execution classes satisfying properties 1 through 4 listed above.
\end{lemma}

\begin{proof} Property 1 is satisfied by construction. Property 3 follows from the fact that $\mathcal{C}_0$ is taken from a real execution in which the input bits of all good processors are 0; a good processor in such an execution deciding on the value 1 would violate the correctness conditions of the protocol.

To prove property 2, we consider how adjacent classes in the chain are generated. A new class is produced when exactly one of two things happen: either the sets $z_i^j$ change, or the input change. When the sets $z_i^j$ change, it is actually only one set $z_{r_\ell-1}$ that changes. This means that all of the sets $Z_i^j$ for $i < r_\ell-1$ are exactly the same for the two adjacent classes. In fact, the sets $Z_i^j$ are exactly the same for all $i \leq E$ for all $j \neq g_\ell$, since we have made sure that there are no sender, round number pairs with senders in group $g_\ell$ and round numbers $\geq r_\ell$ in any of these $Z_i^j$'s. Since the input bits are unchanged, this means that all of the sets $S_i^j$ for $j \neq g_\ell$ will be the same for the two adjacent classes. If instead it is the input bits of a group $j'$ that have changed between the adjacent classes, then the sets $z_i^j$ and $Z_i^j$ are the same, and none of the $Z_i^j$'s for $j \neq j'$ include any senders from group $j'$. Thus, the change in this group's inputs does not affect the behavior of processors outside the group through the first $E$ rounds, and we have sets $S_i^j$ which are identical for all $i$ whenever $j \neq j'$. This establishes property 2.

Now we must prove termination and property 4. First, we note that termination implies property 4. To see this, consider the termination condition. Termination occurs precisely when the list $\mathcal{L}$ is equal to $\{(n/t, 1)\}$ and the final $E$-round lockstep execution class produced (call this $\mathcal{C}_L$) has input values of 1 for all the members of the final group. To reach this point, the algorithm must have gone through calls where the list was of the form $\{(j,1)\}$ for each $j$ from 1 to $n/t$. The only way for the algorithm to get from a list of $\{(j,1)\}$ to a list of $\{(j+1,1)\}$ is to produce a sequence of intermediary $E$-round lockstep execution classes which begin with the input bits of group $j$ all being 0 and ends with these input bits all being 1. Once these inputs are changed to 1, they are never changed back. Thus, if the protocol terminates after producing $\mathcal{C}_L$, then $\mathcal{C}_L$ must correspond to an execution in which all of the input bits of good processors are equal to 1. Thus, termination implies property 4.

Finally, we prove the algorithm terminates. We consider the way the list $\mathcal{L}$ evolves as the algorithm runs. When a pair is added to the list, its round number is always strictly greater than the round number of the previous list element. These round numbers will never exceed $E+1$. To verify this, note that when the round number of the last list entry $(g_k, r_k)$ is $r_k=E+1$, the condition that no sets $z_i^j$ with $j\neq g_k$ and $i \geq E+1$ include group $g_k$ is trivially satisfied, since we end at round $E$. Thus, the list cannot grow at this point, and instead we are guaranteed to remove its last element. We have thus shown that pairs with round number $E+1$ are always guaranteed to be removed from the list.

We now employ induction on the round number. Suppose that at some point during the running of the algorithm, we have a list whose last pair has round number $r$, and all pairs with round numbers strictly higher than $r$ are guaranteed to be eventually removed from the list. One of two things can happen next: either we will remove the last pair with round number $r$, or we will add a new pair to the list with round number greater than $r$. The round number of this new pair, which we will call $r'$, is chosen to be minimal. We let $g_{\ell}$ denote the group number of this new pair. By the inductive hypothesis, we know that this new pair will eventually be removed from the list (since $r' >r$). From the time that we added the pair with round number $r'$ to the point when we remove it, all of the sets $z_i^j$  for rounds $ i < r'-1$ remain unchanged, and for $i = r'-1$, the only set that changes is $z_{r'-1}^{g_\ell}$.

Next, we will either remove the pair with round number $r$ from the list, or we will add a new pair with the new minimal round number, $r''$. Since we have left the sets $z_i^j$ for all $j$ and all $i < r'-1$ unchanged, this new minimal round number must satisfy $r''\geq r'$, and if $r'' = r'$, the new group number cannot be equal to $g_{\ell}$, since we have already fixed the set $z_{r'-1}^{g_\ell}$. As we continue adding and removing new pairs with round numbers $\geq r'$, this fix will not be undone. Hence, since there are a finite number of groups, we will eventually progress to a point where the new minimal round number is $> r'$. This minimal round number will continue increasing upward, but it cannot exceed $E+1$. This means that at some point, we will remove the pair with round number $r$.

We may conclude that all pairs added to the list are eventually removed. Applying this to the pairs which are added with round number equal to 1, we see that the inputs for each group will eventually be changed from 0 to 1, which guarantees that the process will terminate, with all inputs equal to 1.
\end{proof}

\subsection{Completing the Proof of Theorem \ref{thm:main}}
We consider our chain of $E$-round lockstep execution classes $\mathcal{C}_0, \ldots, \mathcal{C}_L$ satisfying properties 1 through 4. Among these, there is some $\mathcal{C}_{\ell^*}$ which results in some good processors remaining undecided after $E$ rounds. We now use Lemma \ref{lem:probability} to complete our proof. We recall that this lemma shows that when the good processors in a group each sample their next message independently from the same distribution $\mathcal{D}$ on at most $R$ possibilities, with high probability the adversary can choose the ``random" bits of the faulty processors in the group to ensure that the number of times each possible message is chosen within the group exactly matches the expected number under distribution $\widetilde{\mathcal{D}}$.

We let $\delta'$ denote the value $\delta c^3/(3(1-c))$ appearing in the statement of Lemma \ref{lem:probability}. We note that $\delta'$ is a positive constant which can be chosen to depend only on $R$ and $c$ (recall the $\epsilon$ is chosen with respect to $R$ and $c$). We consider an execution which begins with the same input bits as $\mathcal{C}_{\ell^*}$. As the execution runs, the adversary will choose the message scheduling and the supposedly random bits for the faulty processors in an attempt to create message sets through the first $E$ rounds that match the sets $S_i^j$ associated to $\mathcal{C}_{\ell^*}$. We note that the scheduling can be chosen according to permutations $\pi_{p,i}$ for each processor $p$ and each round $i$ derived from the sets $Z_i^j$ for the class $\mathcal{C}_{\ell^*}$ as described in Definition \ref{def:permutations}. More precisely, in each round $i$, we run (or continue running) parts of the finite schedules $\sigma_{\pi_{p,i'}}$ for all $i' \leq i$ and stop when exactly those processors in groups $j$ with $(p, i') \in Z_i^j$ have accepted each message.

In order for the adversary to be successful in creating an execution that falls into class $C_{\ell^*}$, it must ensure that the messages chosen in each round by each group conform precisely to the expected numbers for each possibility under the corresponding distribution $\widetilde{\mathcal{D}}$. This can be done as long as the number of good processors in the group choosing each possibility $s$ do not exceed the expected number, $\tilde{\rho_s}t$. When this occurs, the adversary can set the messages of the faulty processors in the group so that each expectation is matched precisely. We note that the sets $S_i^j$ always contain \emph{all} of the round $k$ messages sent by a group or none of them (recall we have required that if $(p,k) \in Z_i^j$ for any processor $p$, any rounds $i,k$, and any group $j$, then $(p',k) \in Z_i^j$ for all processors $p'$ in the same group as $p$). Thus, as long as the adversary achieves the desired multi-set of messages for each group, the sets $S_i^j$ of $\mathcal{C}_{\ell^*}$ will be attained. (It does not matter which processor from each group sends which message, as long as the multi-set of messages produced by each group matches the specification of $\mathcal{C}_{\ell^*}$.)

Since there are at most $R$ possible messages for each group and there are $n/t = 1/c$ groups, the union bound in combination with Lemma \ref{lem:probability} ensures that the probability of the adversary failing in any given round is at most $\frac{R}{c} e^{-\delta' n}$. Thus, the adversary will succeed in producing the sets $S_i^j$ associated with $\mathcal{C}_{\ell^*}$ through $E$ rounds with probability at least $1 - \frac{ER}{c} e^{-\delta'n}$. When the adversary succeeds, some good processors will remain undecided at the end of $E$ rounds.

We now fix the value of $E$ as:
\[ E := \frac{c}{2R} e^{\delta'n}.\]
This is exponential in $n$, and the probability that the adversary can force the execution to last for at least $E$ rounds is $\geq \frac{1}{2}$. This proves that the expected running time is exponential.
This completes our proof of Theorem \ref{thm:main}.
\end{proof}

\section{Directions for Future Work}
We have proven that for any fully symmetric round protocol, there are some input values and an adversarial strategy that will force the execution to run for an exponential number of rounds with constant probability. This results in an exponential expected running time in general for values of $t$ which are linear in $n$.
Our work leaves many interesting open questions and illuminates several potential directions for future work on understanding the range of possible behaviors for randomized Byzantine agreement algorithms in the asynchronous, full information setting. We hope that the restrictions we placed on fully symmetric round protocols in order to implement our proof strategy may provide useful clues for where one should look when searching for polynomial expected time algorithms (particularly Las Vegas algorithms). Informally speaking, we may ask: how far does one have to go beyond the realm of fully symmetric round protocols in order to find an expected polynomial time algorithm? Does one have to abandon symmetry completely? Or might one deviate from our specifications in more subtle ways?

\paragraph{Weaker Symmetry} For instance, we could consider an enlarged class of protocols that is symmetric in a weaker sense: behavior could still be invariant under permutations of the processor identities attached to accepted messages, but these permutations could be fixed for the entire history of previous rounds, instead of allowed to change per round. We do not know whether our impossibility result can be extended to protocols exhibiting this weaker kind of symmetry.

\paragraph{More Randomness} It is also intriguing to consider the small change of lifting the restriction on the number of random choices. Though our probabilistic analysis is not nearly optimized, it does seem fairly sensitive to the number of possibilities considered when a processor makes a random choice. Having more choices will considerably decrease the adversary's chances of arranging the numbers of all outcomes to conform with their adjusted expectations. However, it is not clear how to leverage using more randomness to achieve faster Las Vegas algorithms.

\paragraph{Use of Round Structure and Primitives}
It is also worth considering if the seemingly natural notion of a round (imported from the synchronous setting) may have a restrictive effect on our thinking in the asynchronous setting. Chor and Dwork describe the nature of an asynchronous round as follows: ``...there is something akin to a round even in a completely asynchronous system. Consider a set of $n$ processors running a protocol tolerant to $t$ faults, and let $p$ be a correct processor in this set. If $p$ broadcasts the message for its $i^{th}$ step in the protocol and receives step $i$ messages from only $n-t$ processors, then $p$ cannot safely wait for additional step $i$ messages because all $t$ processors from which it has not heard may be faulty. In this case, $p$ must proceed to step $i+1$ in its protocol, and we say $p$ has completed round $i$" \cite{CD89}. This reasoning is convincing, but also a bit deceptive. For protocols like Ben-Or's \cite{BO83} and Bracha's \cite{B84} which consult only the current round's messages, this reasoning is sound, but for protocols which may use more of the execution history, the situation is more subtle. For example, consider a processor who is in round 2 and has received round 2 messages from processors 1 through $n-t$. Suppose it previously received round 1 messages from processors $t+1$ through $n$. Then there are a total of $2t$ processors that it has failed to hear from so far, so it may safely wait to receive $t$ more messages for rounds 1 and 2 combined before moving to round 3. We have allowed for this sort of behavior in our validation primitive, but there could be other subtle violations of our restrictive notion of round behavior that could allow protocols to avoid our impossibility result.

This issue is related to our requirements for the Broadcast and Validate primitives. It is possible that one might leverage instantiations of these primitives with stronger properties or employ wholly new primitives to avoid our result without acquiring considerably more complexity in the high-level algorithm.

\bibliographystyle{plain}
\bibliography{BA}

\appendix
\section{Proof of Lemma \ref{lem:probability}} \label{app:chernoff}
\begin{proof} We first consider $s \neq s^*$. In this case, $\rho_s \leq \tilde{\rho_s}$, and $\tilde{\rho_s} \geq \epsilon$. We define new independent random variables $\widetilde{X}_1, \ldots, \widetilde{X}_{(1-c)t}$ which are equal to 1 with probability $\tilde{\rho_s}$ and equal to 0 with probability $1 - \tilde{\rho_s}$. Since $\rho_s \leq \tilde{\rho_s}$, we have that:
\[\mathbb{P}\left[ \sum_{i=1}^{(1-c)t} X_i \geq \tilde{\rho_s} t \right] \leq \mathbb{P}\left[\sum_{i=1}^{(1-c)t} \widetilde{X}_i \geq \tilde{\rho_s} t\right].\]

We note that
\[\tilde{\rho_s} t = \tilde{\rho_s} (1-c)t \left(\frac{1}{1-c}\right) = \left(\frac{1}{1-c}\right) \mathbb{E}\left[ \sum_{i=1}^{(1-c)t} \widetilde{X}_i\right].\] Since $\frac{1}{1-c}  = 1 +\frac{c}{1-c}$ and $\frac{c}{1-c} <1$, the Chernoff bound yields:
\[\mathbb{P}\left[ \sum_{i=1}^{(1-c)t} X_i \geq \tilde{\rho_s} t \right] \leq e^{-\tilde{\rho_s} c^2 t/(3(1-c))} \leq e^{- \epsilon c^3n/(3(1-c))}.\]

We now consider $s^*$. Then $\rho_{s^*} \geq \frac{1}{R}$, and $\widetilde{\rho}_{s^*} \geq \rho_{s^*} - R(\epsilon+ \frac{1}{t})$. We then have:
\[\mathbb{P}\left[ \sum_{i=1}^{(1-c)t} X_i \geq \tilde{\rho}_{s^*} t \right] \leq \mathbb{P} \left[\sum_{i=1}^{(1-c)t} X_i \geq \left(\rho_{s^*} - R\left(\epsilon + \frac{1}{t}\right)\right) t\right] .\]

We can rewrite $\left(\rho_{s^*} - R\left(\epsilon + \frac{1}{t}\right)\right) t$ as
\[\rho_{s^*}(1-c)t\left(\frac{1}{1-c}\right)\left(1- \frac{R(\epsilon +1/t)}{\rho_{s^*}}\right).\] Since $\rho_{s^*} \geq \frac{1}{R}$, this quantity is:
\[\geq \left(\frac{1}{1-c}\right) \left( 1 - R^2(\epsilon + 1/t)\right) \mathbb{E}\left[ \sum_{i=1}^{(1-c)t} X_i\right].\]
Since $\epsilon$ was chosen so that $\epsilon + \frac{1}{t} < \frac{c}{2R^2}$, we have
\[\left(\frac{1}{1-c}\right) \left( 1 - R^2(\epsilon + 1/t)\right) \geq \left(\frac{1}{1-c}\right) \left(1 - \frac{c}{2}\right) = 1+ \frac{c}{2(1-c)}.\]
Hence, by the Chernoff bound (since $0 < \frac{c}{2(1-c)} \leq 1$), we have:
\[\mathbb{P}\left[ \sum_{i=1}^{(1-c)t} X_i \geq \tilde{\rho}_{s^*} t \right] \leq e^{-(1-c)cn/3\cdot \frac{c^2}{4R(1-c)^2}} = e^{-c^3 n/(12R(1-c))}.\]

Since $\delta := \min \{\epsilon, \frac{1}{4R}\}$, we have shown that
\[\mathbb{P}\left[ \sum_{i=1}^{(1-c)t} X_i \geq \tilde{\rho}{_s}t\right] \leq e^{-\delta  c^3 n/(3(1-c))}\]
holds in all cases.
\end{proof}

\end{document}